\pdfoutput=1
\documentclass[10pt,conference,letterpaper]{IEEEtran}
\usepackage{times,amsmath,epsfig}
\usepackage{balance}
\usepackage{caption}
\usepackage{hyperref}
\usepackage{url}
\usepackage{microtype}
\usepackage{paralist}
\usepackage{booktabs}
\usepackage{amssymb}
\usepackage{mathtools}
\usepackage{listings}
\usepackage{cite}
\usepackage{subscript}
\usepackage{stmaryrd}
\usepackage{pifont} 

\usepackage{amsthm}
\usepackage{float}
\usepackage[algoruled, linesnumbered, vlined]{algorithm2e}
\usepackage[utf8]{inputenc}

\usepackage{prp-macros}

\usepackage{setspace}

\makeatletter
\newcommand\subparagraph{%
  \@startsection{subparagraph}{5}
  {\parindent}
  {3.25ex \@plus 1ex \@minus .2ex}
  {-1em}
  {\normalfont\normalsize\bfseries}}
\makeatother
\usepackage{titlesec}
\let\subparagraph\relax 

\usepackage{titlesec}

\titlespacing*{\section}
{0pt}{.5ex plus 1ex minus .2ex}{0ex plus .2ex}
\titlespacing*{\subsection}
{0pt}{.5ex plus 1ex minus .2ex}{0ex plus .2ex}
\titlespacing*{\subsubsection}
{0pt}{.5ex plus 1ex minus .2ex}{0ex plus .2ex}

\newtheorem*{example*}{Example}
\newtheorem{theorem}{Theorem}[section]

\newtheorem{lemma}[theorem]{Lemma}

\lstdefinestyle{myCSharp}{
   language={[Sharp]C},
   basicstyle=\scriptsize,
   belowskip=-2em,
   backgroundcolor=\color{white},   
   basicstyle=\scriptsize\ttfamily,        
   breakatwhitespace=false,         
   breaklines=true,                 
   captionpos=b,                    
   commentstyle=\color{Brown},    
   deletekeywords={...},            
   escapeinside={\%*}{*)},          
   extendedchars=true,              
   frame=,                    
   keepspaces=true,                 
   keywordstyle=\color{MidnightBlue},       
   morekeywords={var},            
   numbers=none,                    
   numbersep=5pt,                   
   numberstyle=\tiny\color{gray}, 
   rulecolor=\color{black},         
   showspaces=false,                
   showstringspaces=false,          
   showtabs=false,                  
   stepnumber=1,                    
   stringstyle=\color{RedOrange},     
   tabsize=2,                       
   title=\lstname,                   
   moredelim=[is][\color{blue}\bfseries\underbar]{@}{@},
}

\lstdefinestyle{mySQL}{
   language={SQL},
   basicstyle=\scriptsize,
   belowskip=-2em,
   backgroundcolor=\color{white},   
   basicstyle=\scriptsize\ttfamily,        
   breakatwhitespace=false,         
   breaklines=true,                 
   captionpos=b,                    
   commentstyle=\color{Brown},    
   deletekeywords={...},            
   escapeinside={\%*}{*)},          
   extendedchars=true,              
   frame=,                    
   keepspaces=true,                 
   keywordstyle=\color{MidnightBlue},       
   morekeywords={var},            
   numbers=none,                    
   numbersep=5pt,                   
   numberstyle=\tiny\color{gray}, 
   rulecolor=\color{black},         
   showspaces=false,                
   showstringspaces=false,          
   showtabs=false,                  
   stepnumber=1,                    
   stringstyle=\color{RedOrange},     
   tabsize=2,                       
   title=\lstname,                   
   moredelim=[is][\color{blue}\bfseries\underbar]{@}{@},
}

\graphicspath{{images/}}

\def\Snospace~{\S{}}

\DeclareMathOperator*{\argmax}{arg\,max}
\DeclareMathOperator*{\argmin}{arg\,min}

\SetCommentSty{mycommfont}

\newcommand{\ds}{\textsc{Xtructure}\xspace}
\newcommand{\system}{\textsc{Xsystem}\xspace}

\newcommand{\CM}[1]{}
 \newcommand{\SG}[1]{}
 \newcommand{\LRB}[1]{}
 \newcommand{\srm}[1]{}

\newcommand{\eat}[1]{}

\renewcommand{\ldots}{\ifmmode\mathinner{\ldotp\kern-0.1em\ldotp\kern-0.1em\ldotp}\else.\kern-0.13em.\kern-0.13em.\fi}

\begin{document}


\title{Extracting Syntactic Patterns from Databases}



\author{%
{Andrew Ilyas, Joana M. F. da Trindade, Raul Castro Fernandez, Samuel Madden}%
\vspace{1.6mm}\\
\fontsize{10}{10}\selectfont\itshape
\,CSAIL, MIT
\fontsize{9}{9}\selectfont\ttfamily\upshape
$<$aiilyas,jmf,raulcf,madden$>$@csail.mit.edu
}


\maketitle


\begin{abstract}
 
Many database columns contain string or numerical data that conforms to a
pattern, such as phone numbers, dates, addresses, product identifiers, and
employee ids. These patterns are useful in a number of data processing
applications, including understanding what a specific field represents when
field names are ambiguous, identifying outlier values, and finding similar
fields across data sets.

One way to express such patterns would be to learn regular expressions for each
field in the database.  Unfortunately, existing techniques on regular
expression learning are slow, taking hundreds of seconds
for columns of just a few thousand values. In contrast, we develop
\system, an efficient method to learn patterns over database columns in 
significantly less time.

We show that these patterns can not only be built quickly, but are expressive
enough to capture a number of key applications, including  detecting outliers,
measuring column similarity, and assigning semantic labels to columns (based on
a library of regular expressions).  We evaluate these applications with datasets
that range from chemical databases (based on a collaboration with a
pharmaceutical company), our university data warehouse, and open data from
MassData.gov.

\end{abstract}

\section{Introduction}
\label{sec:introduction}

Modern enterprises store their data in a wide range of different systems,
including transactional DBMSs, data warehouses, data lakes, spreadsheets, and
flat files. Data analysts often need to combine data from these diverse data
sets, frequently incorporating external data from even more sources.  A key
challenge in this setting is finding related data sets that can be combined to answer some
question of interest.

As an example, analysts at Merck---a pharmaceutical company---often need to
join  tables that contain chemical compounds. Unfortunately, there are at
least three identifier formats (\eg InChI, InChIKey, and SMILES, shown in
\F\ref{fig:compoundids})  used internally in Merck, not to mention
additional formats that may be used in external data sources. Because of this
diversity of ID formats, a simple text search is not sufficient to find relevant
tables---attribute names are different. Indeed, they cannot even perform an
approximate search to find similar content as these identifiers are not
comparable.  Manually building a mapping between the identifiers in the
different formats and creating a lookup table is an expensive option.


\begin{figure}[h]
\centering
\includegraphics[width=0.7\columnwidth]{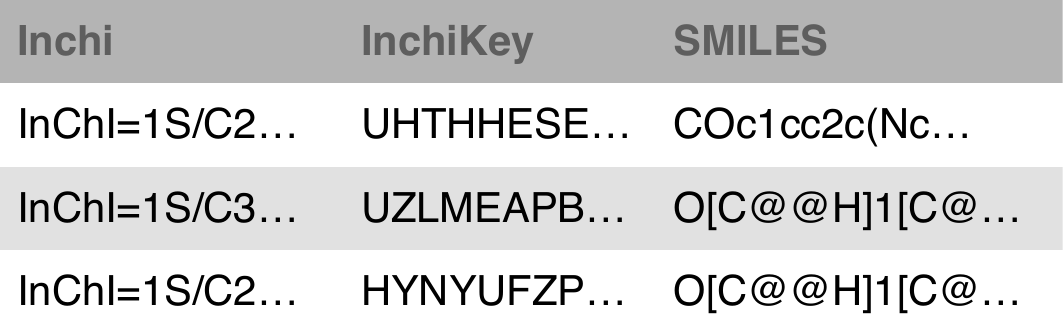}
\caption{Example of different chemical compound ID formats}
\label{fig:compoundids}
\end{figure}

A better option would be to  label the relevant attributes with useful
metadata, \eg assign a \emph{chemical identifier} label to all identifier columns in the table
that represent. Unfortunately, manual labeling is also infeasible in an company
with large volumes of data: it requires a great deal of time and is error prone
as it may miss many tables that contain relevant information, especially when
considering external data.

To address this problem, we observe that many relevant attributes in enterprise
databases are \emph{highly structured}, \ie they follow simple syntactical
patterns.  For example, in \F\ref{fig:compoundids} the InChi number always starts
with the pattern \emph{InChI=} and the InChiKey is a 14-character followed by a
hyphen, followed by a 10-character followed by another hyphen and an additional
character \cite{inchikey}.  More common examples of structured attributes are
dates, product identifiers, phone numbers, enumerated types (gender, etc), and so on. 
Often these columns are stored as strings in the database, but if they could be
labeled with richer structural information about the format of values, indexing, searching and comparing values,
and finding exceptional or outliers values, could be done much more efficiently.

In this paper, we introduce \textbf{\system}, a method to learn and represent
syntactic patterns in datasets as data structures called \ds{s}. Once \system learns
a collection of patterns, analysts can use them to conduct several commonly
performed tasks, including: {\it automatic label assignment},
where data items are assigned a class by comparing them to a library of known
classes (written as regexs or \ds{s}); {\it finding syntactically similar content},
where learned \ds{s} are compared to see if they are similar, and {\it outlier detection},
where a learned \ds for a single item is compared to other \ds{s} to check that its structure
is different.  These applications share two common requirements: (i) \ds{s} must be quickly synthesizable and (ii) \ds{s} must be comparable to each other and to regular expressions.

In addition to supporting these requirements, \system must: i) be able to work without human
intervention, as neither semi-automatic nor interactive tools scale for
large amounts of data; ii) learn syntactic patterns fast, which calls
for both an asymptotically efficient model as well as a parallelizable
implementation; and iii) be quickly synthesizable and manipulatable given only
raw datasets, since this is all that many analysts may be able to initially access.
Speed of learning is crucial for real world scenarios, as not all data analysis
tasks can cope with stale data.

\mypar{Other Methods} A natural question to ask is why not to use existing methods
for information extraction that different communities have been improving over
the years? We discuss methods in detail in the related work section, but in
general we find that these systems either lack speed, lack autonomy, or make
assumptions about the problem setting that are too strong for realistic use.
One general method that seems relevant for structured data is
inference of regular expressions \cite{FERNAU2009521} to represent the structure of
each column. This is unfortunately NP-hard, due to the expressiveness of the
model \cite{gold}. We argue, and show through many examples, that regexes'
expressive power is not necessary for syntactic-based discovery applications in
databases. This is the key observation that we use to design a less-expensive,
much faster-to-learn structure which can express a subset of regular expressions
that satisfies most applications.  We show that existing regex learning
algorithms are impractically slow, taking thousands of seconds to learn a regex
over just a few hundred values.
 Other general methods to extract patterns grow out of the need for
identifying important entities within unstructured or semi-structured data, such
as NLP-based techniques to identify patterns from text or wrapper induction
techniques to do the same on webpages built on HTML. None of those techniques
are designed for the case of structured data, which has become a big bottleneck
in the era of big data, with people placing a variety of structured data into data
lakes. 

\newcommand{\specialcell}[2][c]{%
  \begin{tabular}[#1]{@{}c@{}}#2\end{tabular}}

\mypar{\system}Our approach learns syntax from examples incrementally. For each example, it exploits
the existence of delimiters in known entities to split the problem of extracting
the pattern into learning the syntax of each of the \emph{tokens} separated by
those delimiters. The underlying data structure used to learn each token is a branching linear distribution
sequence that is equivalent to a Deterministic Acyclic Finite State Automaton (DAFSA),
which is asymptotically simpler to learn than minimal Deterministic Finite Automata (DFA), often used
in regular expression learning. The learning procedure relies on a \emph{branch and merge} strategy
that allows us to incrementally adapt a prior to new observed examples.  This
permits us to capture different syntactical structures that appear in the same
column. This branch and merge strategy is also at the center of the
parallelization approach used in \system.


We evaluate \system on the three applications mentioned above on real datasets
ranging from our university's data warehouse, open government data and a public
chemical database. We find that \system can form a syntactic representation of
given data much faster than automatic DFA learners, and that we can use it
effectively for our target applications.

Before describing the details of our implementation, we now formalize the
\system problem statement and requirements.

\section{Motivation and Requirements}
\label{sec:background}

In this section, we first restate and elaborate on the applications that
motivated us to build \system, and then use those applications to derive a set
of requirements, which we use to derive the features of our new system.

\subsection{Application Scenarios}
\label{subsec:requirements}

We focus on three applications of \system that we have identified while working
with collaborators in  pharmaceutical, telecommunications and data
integration applications.

\vspace{.05in}
\noindent {\bf 1. Automatic Label Assignment.} 
Automatic label assignment attaches a semantic type (e.g. ``chemical compound ID'', ``phone number'')
to columns in a data set, so that users can understand the content of
columns and perform semantic search for similar types of columns.
A key observation is that many different semantic types are already available in regex libraries~\cite{regexlib}, 
and for important semantic types inside an organization, writing such a regex is
relatively straightforward.  For example, for the examples of 
\F\ref{fig:compoundids}.

Given such a table of (regex, semantic label) pairs,
the idea in automatic label assignment is to learn a \ds for each attribute in the database,
and then perform a search for similar regexes in this table, to assign a semantic label
to each column in the database.  This introduces two key requirements for \ds{s}:  1) they must be fast
to learn, since we need to infer them for every column in the database, and 2) they must be comparable to 
regular expressions.

\vspace{.05in}
\noindent{\bf 2. Summarization and Attribute Comparison.} Once some interesting
attributes are identified, data analysts often wish to find other similar
attributes across datasets (e.g., to
obtain candidates for joining two datasets together.) One way to achieve this
would be to compare  every pair of attributes in each dataset to each other.
However, a naive quadratic implementation would be prohibitively expensive, which has inspired 
 many approaches based on set-similarity joins and approximate
methods (e.g., \cite{setjoin1, setjoin2}).

\begin{figure}[t]
\centering
\hspace{-0.7cm}
\includegraphics[width=0.5\columnwidth]{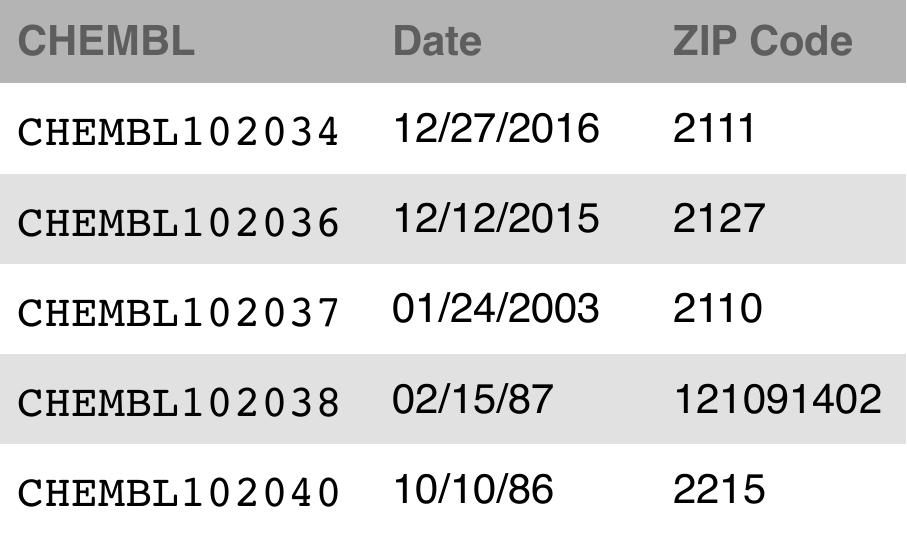}
\caption{Examples of attributes that \system might learn}
\label{fig:structured}
\end{figure}


A complementary approach that we advocate 
is to learn a more compact representation, e.g., a \ds for each attribute in a database.
We can then compare these representations
instead of the raw data, which offers several benefits. First, the \ds is
interpretable by humans, helping to identify content quickly, while requiring
only a fraction of the space used by the original data. For example, in the case
of the data of \F\ref{fig:structured} it is possible to represent the CHEMBL id
with a single pattern that entirely captures the structures. Second, if the two
\ds{s} can be compared directly, then that similar content can be found without
performing I/O to read data from
the source database. 
Last, the cost of learning the \ds is paid only once,
and can be reused subsequently for other applications as we are describing in
this section.


To be useful for summarization and comparison, \system must learn human-readable {\ds}s, similar to regexes in common programing languages, and \ds{s}
must be comparable to one another, to permit finding similar content.

\vspace{.05in}
\noindent{\bf 3. Syntax-Based Outlier Detection.} One long-standing problem in data
management is concerned with data quality; in particular, errors occur 
frequently, whether due to data entry or anomalous values or readings~\cite{errors}.


We observe that by learning the syntactic pattern of an attribute, we can
detect many types of errors, particularly those that are {\it syntactic
outliers}, \ie elements that do not closely the match a learned \ds.
Consider the example of \F\ref{fig:structured}, with real ZIP codes from Boston.
The 4th cell value an erroneous ZIP code. In this case, it is possible to detect
that it has a different length than other records in the same attribute and
 does not fit the general syntactic pattern of the column. 

To be able to detect syntax-based outliers, \ds{s} must support the concept
of a \emph{scoring fit}, \ie a numeric score capturing how well a value fits a learned \ds. Also,
\system to be used as an outlier detector, it must not overfit the \ds to all
the values, or it will not detect outliers. Instead, it must represent the
general syntactical pattern and not capture the content of a few outliers. 

\subsection{Summary of Requirements}

Based on the previous applications, we can summarize \system's requirement for \ds{s} as follows:

$\bullet$ \textbf{Comparable to Each Other.} We need to be able to compare \ds{s} to each other.
 Intuitively, we want the distance between two \ds{s}
$X_i$ and $X_j$ to approximate some sort of ``average distance" between the
domains of possible values that $X_i$  $X_j$ range over. 
This is useful in our summarization and comparison application.

$\bullet$ \textbf{Comparable to DFAs/Regular Expressions.} To support label assignment, we need
to be able to compare a \ds{s} to a regular expression.

$\bullet$ \textbf{Able to Quantify Fit.} Not only should \ds{s} be comparable, 
but we should be able to efficiently compute a numeric score that captures the goodness of fit of
two \ds{s}, or of a \ds to a regular expression.

$\bullet$ \textbf{Quick to Learn.} \system should be able to learn \ds{s}
efficiently and in parallel, to support applications that need to compare large
numbers of \ds{s} to each other or to regular expressions.

\subsection{Motivation for a New Approach}

It may seem that the above requirements could be trivially satisfied by learning
regular expressions (DFAs) over each column.  However, as described in more
detail in our related work (Section ~\ref{sec:relatedwork}), regular expression
learning~\cite{relie,angluin,349918,FERNAU2009521,Brazma:1993:EIR:168304.168340,Lee:2016:SRE:2993236.2993244,Bartoli:2014:ASR:2780227.2780354,Bartoli:2016:IRE:2925263.2925390,Brauer:2011:EIE:2063576.2063763}
is in general an NP-complete problem, and in practice solutions for finding
regular expressions are extremely inefficient. 

\begin{figure}[t]
\vspace{-.2in}
  \centering
      \includegraphics[width=0.7\columnwidth]{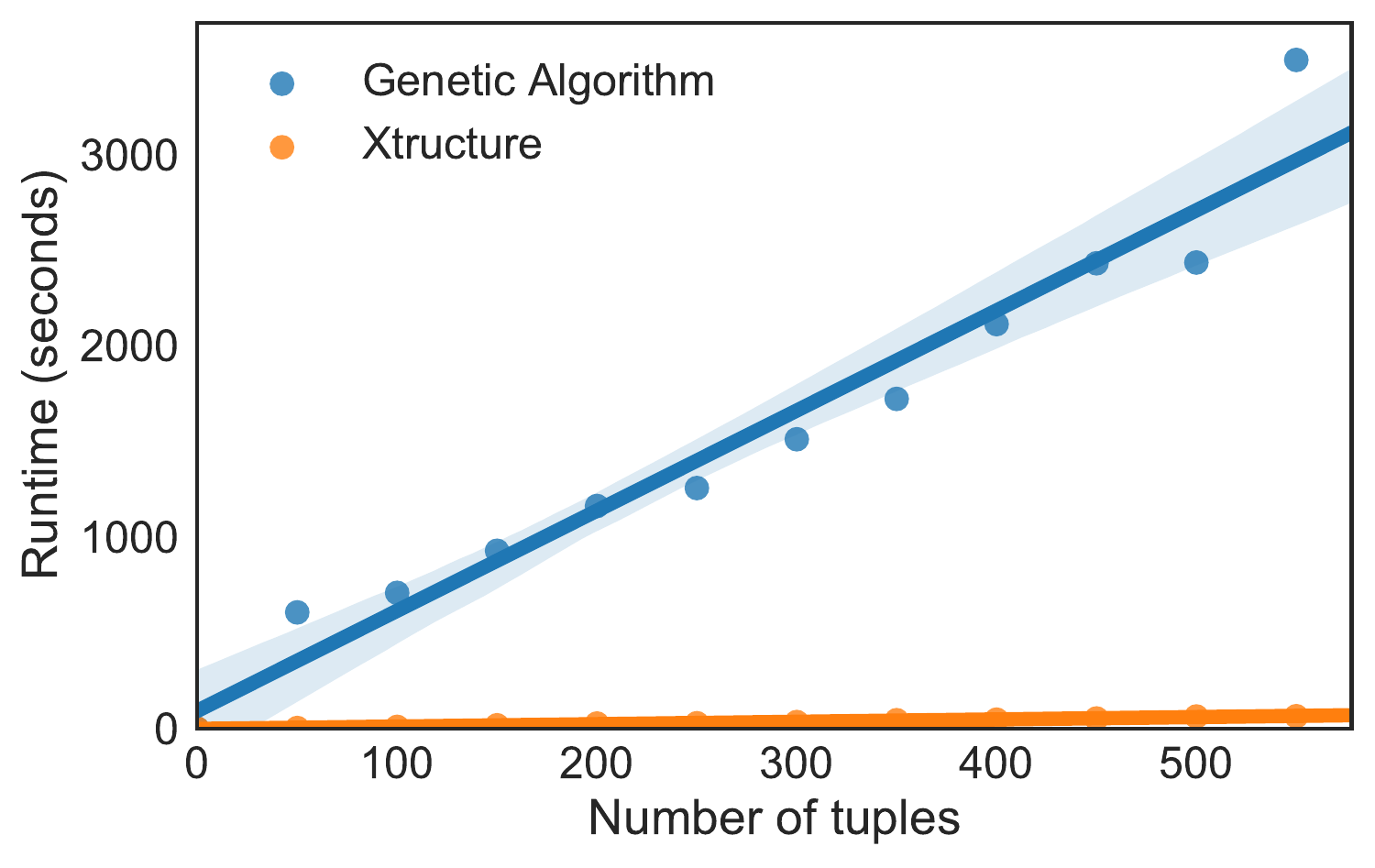}
  \caption{\system runtime vs. full regex learning algorithm~\cite{Bartoli:2016:IRE:2925263.2925390}}
  \label{f:runtime}
\end{figure}

\mypar{How inefficient is to learn regex?} To build intuition about this
inefficiency, we used a state-of-the-art regex inference
algorithm~\cite{Bartoli:2016:IRE:2925263.2925390} to learn a regex over a few
hundred tuples and found that it took around an hour to complete.
Figure~\ref{f:runtime} shows the speed of learning a \ds from data using \system
with the state of the art algorithm~\cite{Bartoli:2016:IRE:2925263.2925390}.
Here we show the time to learn a regular expression or a \ds over a column, as
the length of the column (in tuples) grows.  The genetic algorithm based method
is infeasible for our target applications because it takes thousands of seconds
to learn a regular expression for a single column, making it impractical to use
in even a moderate collection of databases with a few hundred columns.  In
contrast, the performance of \system with \ds{s} grows sub-linearly with the
number of tuples, as we will show in subsequent sections.

If regular expressions were available, we could use them to solve the
application scenarios we showed above. However, because regular expression
learning algorithms solve a more complex problem than what is needed for the
applications we have identified at a high computational cost, we sought a
 simpler language that is both efficient to learn and that is
sufficient to capture the structure of many database columns.

\mypar{The Opportunity} Fortunately, we have observed that real data in databases is often quite simple,
and does not require the full expressivity of DFAs/regular expressions.  In
particular, most attributes in database have the following properties:

\begin{figure}[t]
\vspace{-.2in}
\centering
\hspace{-0.5cm}
\includegraphics[width=0.7\columnwidth]{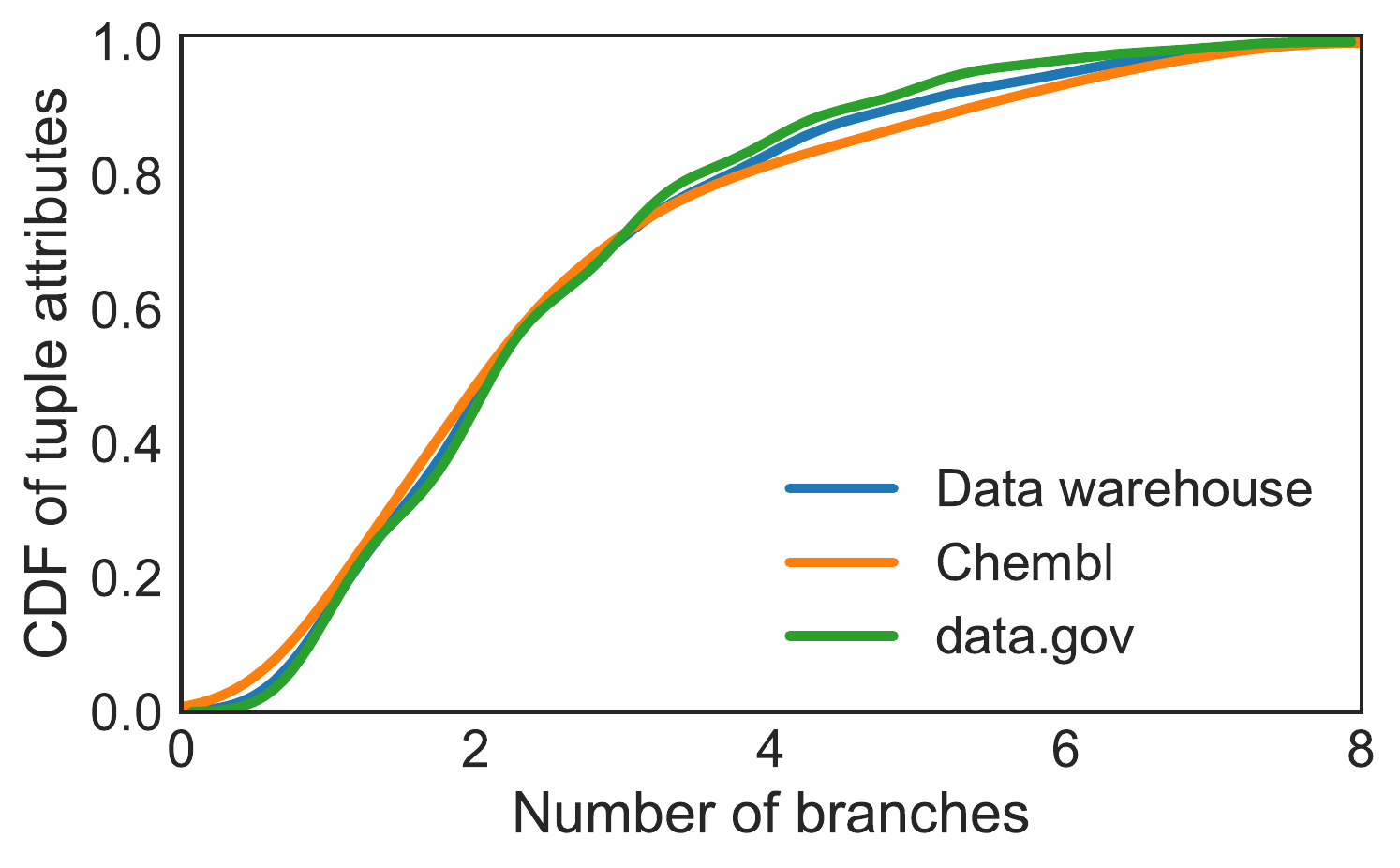}
\caption{Number of \ds branches per tuple attribute (column) for different datasets.}
\label{fig:branchescdf}
\end{figure}

$\bullet$ \textbf{Simple structure.} Through the wildcard ``*" and ``+" operators,
regexes allow infinite variability of structure within a domain. In practice, on
the MassData dataset (open data from Massachusetts), we found that around 20\%
of columns are fixed length,  over half have only 3 distinct
column lengths,  more than 85\% have average length less than 10, and 99\% have average length less than 50. This makes sense because
databases are designed to be easy to manipulate and process, and constraining
the data formats into well-structured values helps achieve this goal.
Further, many regular expression learning papers focus on learning
a {\it minimal} regular expressions, but since database columns are already simple,
 minimality is not a primary concern, especially if it comes at the cost of efficiency.

$\bullet$ \textbf{Consistent structure.} The optionality operator in regular
expressions allows one to construct concise expressions such as ``AB(C)DE."
Instead, the equivalent ``$ABDE|ABCDE|...$", which separates each pattern into a
different \emph{branch} is simpler to learn. We found 40\% attributes of data.gov
can be represented by at most 2 global branches, and nearly 100\% by at most 8.
We show the number of branches required to represent the data of 3 real datasets
in \F~\ref{fig:branchescdf}, confirming the same trend. Again, regular expressions
favor expressivity over efficiency, which isn't necessary.





\notera{I removed the previous paragraph where we were pseudo-formally
introducing Xtructure repr.}

In short, regexes are neither necessary (too expressive) nor
sufficient (they are too slow) for solving the problem of structure learning
addressed in this paper. Instead, as we show, less complex \ds{s} can be learned
 more efficiently while still capturing the structure of 
real databases.

\section{\system Implementation}
\label{sec:approach}

\notera{address small comments from R2.D3}

In this section we introduce the \ds model to learn syntactical patterns from
structured data. 

\begin{figure}[ht]
  \centering
  \includegraphics[width=0.6\columnwidth]{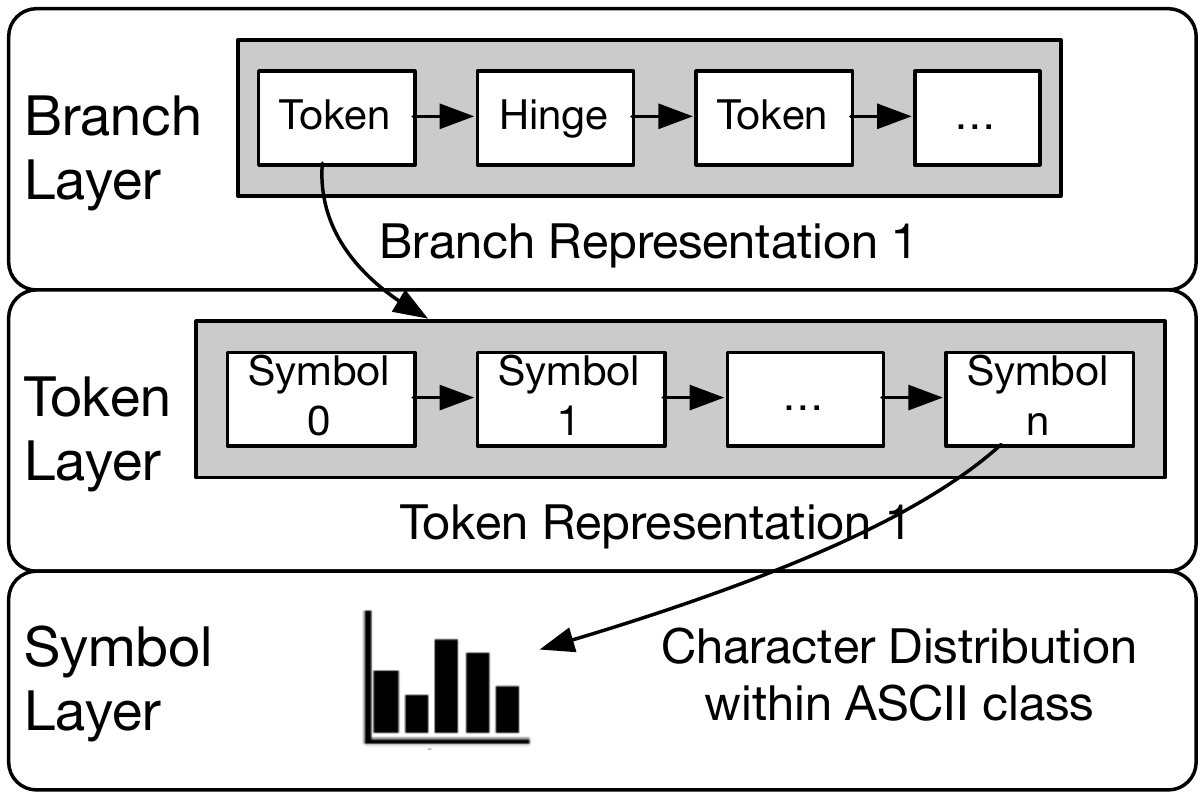}
  \caption{Xtructure data model}
  \label{fig:general_structure}
  \vspace{-.1in}
\end{figure}

%

\subsection{The \ds Model}
\label{subsec:modeloverview}

The goal of \system is to learn a \ds from examples $\{T\} \subseteq
A$ incrementally (tuple by tuple). To do this, we design an architecture that
allows us to probabilistically model each example, and thus at any point output
the ``current" representation.  The architecture of \ds (\F~\ref{fig:general_structure}) has several layers with distributions at the
foundation; tuples are fit into the model by passing them through this layered
structure in a well-defined way. The layers are organized hierarchically, with
each one taking care of a different aspect of the learning process. We explain
each layer's role next.

The bottom layer in the hierarchy is the \textbf{symbol layer}, which holds a
distribution over the ASCII characters that occur at a given position in the
input tuples. This permits us to represent a position in a tuple as a character
class, an or-statement, a single character, or a wildcard (``.") based on the
distribution. For example, if a series of mm/dd/(yy)yy dates are fed to a \ds,
the first character will hold a distribution containing only the values $1$ or
$0$, (since $0 \leq $months$\leq 12$) of the year. The second will eventually
converge to a uniform distribution over $[0,9]$ and thus it will be represented
with the character class \emph{digit}, \emph{D}. We explain how we decide each
representation from later in the paper.

The  \textbf{token layer} represents sequences of characters from the original tuples, 
or {\it tokens}, 
obtained by splitting the original tuple according to a set of
\emph{delimiters}, \eg \emph{-}, \emph{/}, \emph{\#}.  The
intuition is that delimiters 
often capture substructure of tuples. Consider the ``10/1/2017'' date as an
example: here the three tokens are separated by
\emph{/}. Each
token is  represented in a \emph{token representation}, which is simply a
linked list of symbols (from the symbol layer).  For dates, the
``months" in the date will be a token in the token layer, eventually
represented as \emph{(0|1)D}. When no delimiters are available in the data, the
entire string is represented as a single token.

Tokens of different lengths cannot be represented with a
single token layer. The next layer in the hierarchy, called \textbf{branch
layer}, deals with variable-length  data. A branch layer consists of a
list of token layers, and can represent an entire tuple. In particular, a branch
layer represents a list of words -- represented by token structures --
interleaved with delimiters. In our running example, we may find dates with two
different formats for the year, a 4- and a 2-digit one. These two variations
will be represented with two different branches in a \emph{branch representation}.

Each \ds has several branch representations to represent
attributes with different syntactical patterns, for example, tuples with
different lengths. It is common to find dates with many different formats, due
to data quality issues, as well as IDs, capitalization typos, etc.

\subsection{Learning a \ds}
\label{subsec:fitting}

\ds{s} are adapted after each input tuple is consumed. When it gets a new tuple, \system
chooses an existing branch for the tuple, if one exists, or
creates a new branch and seeds it with the input. This decision is made based on
a measure of \emph{scoring fit}. The branch representation then segments the input into
tokens, $K$, based on a set of delimiters, and splits each token, $k$, into
characters, updating the token and symbol layers.

The following sections describe: (1) how we compute  scoring  fit in Section \ref{subsec:scoring}, (2) the
branch-and-merge algorithm to support multiple branches in Section
\ref{subsec:branchandmerge}, (3) the approach to tokenizing input tuples and feeding
characters to the individual layers in Sections \ref{subsec:tokenize} and
\ref{subsec:modeledrep}, and (4)
 an optimization to speed up learning in Section \ref{subsec:earlystop}.

\subsubsection{Fitting Tuples: Scoring Fit}
\label{subsec:scoring}

While learning a \ds $X$, we must understand how well a tuple, $t$, ``fits" into
the structure defined by $X$. We introduce a scoring fit measure for
this. More formally, given a tuple $t$ and a \ds, $X$, we define an operation
$d(t,X): \mathbb{R}^+$ that indicates how far $t$ deviates from the pattern
represented by $X$. This function is useful to fit new examples, as well as to
compare \ds, both to itself and to representations learned with other methods.

To build this function, instead of comparing each character in $t$ to a corresponding
``character" in the representation $X$ (which is ill-defined, since our model 
holds a distribution over characters rather than a single character), we 
look at the characters $S$ a symbol layer \textit{represents}, 
and assign score $d(s_i \in S, l_i)$ for how close each
 $s_i$ matches the representation in symbol layer $l_i$.  To define $d$, 
we use \textsc{Get-Ascii-Class}(c) as the UNIX class  (e.g., alphanumeric, white space, etc) of a character
$c$ , and $l.class$ as the character class of a symbol layer $l$
(referred to as max\_class in Algorithm~\ref{alg:layerrep}). We also define
$l.is\_class$ to be a boolean indicating whether the layer's representation is
its character class. This decision is based on a $\chi^2$ test. If its p-value
cannot be represented as an ``OR" operation over characters, then,
$$d(s_i, l_i) = \begin{cases}
               1 \text{ if }\textsc{Get-Ascii-Class}(s_i) \neq l_i.class\\
               \alpha \text{ if not $l_i$.is\_class and } s_i \not\in l_i.chars\\
               0 \text{ otherwise}
            \end{cases}$$
where $l_i.chars$ are the characters represented in the symbol layer $l_i$, and
the parameter $\alpha$ is used to determine how much exact character matches are 
prioritized compared to matches in class only (i.e two characters).
In practice, we set up $\alpha=\frac{1}{5}$ as a reasonable value for this relative weighting.

We use  $d$  to propagate the symbol layer scoring fit through
a \ds's layers, leading to a general scoring fit of a tuple with respect to the
model. In particular, for token representations $K$, branch representations $B$
and a modeled representation $X$, the distance of a tuple $t$ to $X$ is defined by:
$$d(t, X) = \min_{b \in B} d(t, b)$$ 
that is, the minimum distance of the tuple with one of the branch
representations, $b$, of $X$, which is in turn defined as: 
$$d(t, b) = \sum_{k_i \in b, t_i \in t} d(t_i, k_i)$$
where $t_i$ are the tokens of the input tuple, $t$, that are compared with the
token structures, $k_i$, of $X$ as follows:
$$d(t_i, k) = \left[\sum_{l_i \in k, s_i \in t_i} d(s_i, l_i)\right] + |len(t_i) - len(k)|$$
Note the extra term in the last equation used to pad with null characters
whichever is shorter between the token structure, $k_i$ in $X$ and the token
$t_i$ in the tuple $t$. This ensures \system does not incorrectly penalize smaller valid instances of the underlying finite language,
while still creating a new branch in the structure for them.  For example,
a column that contains several instances of ``123'' and one instance of ``1'',
the latter would be padded with 2 null characters.

\subsubsection{Representing Multiple Branches}
\label{subsec:branchandmerge}

In practice, data from the same attribute may contain values with different syntactical
patterns. For example, an ID might be a 10-digit number, or
simply ``N/A". This phenomena inspires \ds's multiple branch representations 
(that is, why we allow $R$ to be $b-$dimensional). However, we have no way of knowing 
\textit{a priori} how many different patterns are in a set of examples, a \ds
must somehow manage multiple branches, updating and representing them appropriately.

Given a new input tuple, \system must decide whether to fit it into an existing
\ds branch, or create a new branch capture the tuple's syntactical
structure. For this, we use the scoring fit. For each input $t$, \system finds the
``best matching" branch by doing $b_{best} = \argmin_b d(t, b)$; if $d(t,
b_{best})$ is below a \emph{branching threshold}, the tuple is fit into that
branch, otherwise a new ``empty" branch is created. The existence of this
branching threshold introduces the challenge of how to tune it. To avoid manually
tuning such hyperparameter, we introduce an adaptive \emph{branch-and-merge}
technique.

\mypar{Branch-and-Merge algorithm} The algorithm works as follows. We hide the unintuitive and data-dependent
hyperparameter, and instead expose a \emph{maximum
branches} parameter, that indicates the maximum number of structures that are meant to
be represented by a \ds (this is $b$ in the formal definition). 
This parameter can be set up based on domain knowledge,
or user preference, \eg if an analyst knows there are 3  ways of
representing a business entity, he or she can choose 3 as the number of
branches, as no more than those are expected to appear in the data. 

Given a fixed branching threshold and the maximum number of branches desired by
users, \system proceeds as follows: if the number of branches ever exceeds the
specified maximum, then we compute a pairwise distance between branches. The two
closest branches $b_1$ and $b_2$ are merged -- by fitting generated tuples by
the subsumed branch into the one subsuming -- and the new ``branching threshold"
is set to $d(b_1, b_2)$.

This adaptive mechanism allows \system to correct for undershot initial
thresholds, but not overshot ones, so in practice, the initial branching
threshold is set to a small $\epsilon > 0$. The entire algorithm,
including both picking the best branch and branch-and-merge, is shown in further
detail in Algorithm~\ref{alg:highlevelfit}.


\begin{algorithm}[h]
\scriptsize
\SetKwProg{Fn}{Function}{ }{end}
\SetKwData{Branch}{b}
\SetKwData{Branches}{branches}
\SetKwData{BestBranch}{best\_branch}
\SetKwData{Threshold}{branching\_threshold}
\SetKwData{Word}{word}
\SetKwData{MaxBranches}{max\_branches}
\SetKwData{A}{$B_{outer}$}
\SetKwData{B}{$B_{inner}$}
\SetKwData{i}{$B_i$}
\SetKwData{j}{$B_j$}
\SetKwData{BestFit}{$best\_fit$}

\Threshold $\gets \epsilon$ \\
\Fn{learn\_new\_word(\Word: String) : void} {
	\BestBranch$\gets \argmin_{\Branch \in \Branches} \Branch.fit\_score(word)$ \\
	\If{\BestBranch.fit\_score(word) $<$ \Threshold} { 
		\BestBranch.add(\Word)
	} \Else {
		\Branches.add(new Branch(word))
	}
	\If{\Branches.length $>$ \MaxBranches} {
		\tcp{fit(\i, \j) returns how well \i fits into \j}
		\A,\B$\gets \argmin_{(\i, \j) \in \Branches} fit(\i, \j)$ \\
		\Threshold$\gets fit(\A, \B)$ \\
		\A.add\_word(w)$\ \forall\ w \in \B.learned\_words$ \\
		delete \B
	}
}
\caption{Fitting new words into \ds}
\label{alg:highlevelfit}
\end{algorithm}

\smallskip

\mypar{Parallel Learning} One advantage of the \emph{branch-and-merge} algorithm
is that it facilitates parallel learning. When fitting a model,
we can use multiple workers, each one reading disjoint sets of tuples and
fitting them independently. This has the benefit of exploiting the parallelism
readily available in modern architectures, but leads to more than one
representation per attribute. At this point, we can use the
\emph{branch-and-merge} algorithm to merge the branches of the different built
models, leading to a representation equivalent to the one that a single worker
would have learned.

\subsubsection{Tokenization and Character Fitting}
\label{subsec:tokenize}

To update the token layers, the input tuple is split into tokens and then each
token is fed to the layers of its corresponding token structure. The tokenizer
uses special characters (delimiters) as reference for alignment. The positioninig
of these characters on a string is often an indicator of data type. For example,
IPv4 addresses blocks are separated by ``.'', while dates are usually ``/'' or
``-'' delimited.

\subsubsection{Modeled Representation}
\label{subsec:modeledrep}

During modeling, after receiving a new example and
determining the token structures, each token is fed to the layers of its token
structure. A symbol layer, as introduced before, holds a distribution of the
characters it has seen, and represents them with their character class when is
statistically significant (see Algorithm \ref{alg:layerrep}). Each
layer is modeled as a sampling problem, under the hypothesis that every
character within the majority character class is equally likely. A $\chi^2$
test of independence is then performed, confirming or rejecting this hypothesis;
if confirmed, the layer represents itself by its character class (lines
11-12 in Algorithm \ref{alg:layerrep}). If the null hypothesis is rejected, then
there exists  significant bias in the data source that should be captured in
the representation,  so the layer instead enumerates all fit tuples in an
or-statement, in order of decreasing frequency, until a specified ``capture
percentage" of the distribution is captured. This
corresponds to lines 14-20 of Algorithm \ref{alg:layerrep}. Running this process
whenever a new example is encountered ensures we always model a valid \ds.


\begin{algorithm}[h]
\scriptsize
\SetKwProg{Fn}{Function}{ }{end}
\SetKwInOut{Input}{input}\SetKwInOut{Output}{output}
\SetKwData{ClassProportions}{class\_proportions}
\SetKwData{MaxProportion}{max\_proportion}
\SetKwData{MaxClass}{max\_class}
\SetKwData{Class}{class}
\SetKwData{Histogram}{histogram}
\SetKwData{Chars}{chars\_to\_capture}
\SetKwData{AllChars}{all\_chars\_seen}
\SetKwData{CP}{capture\_threshold}
\SetKwData{Captured}{captured}
\SetKwData{Representation}{representation}
\SetKwData{Char}{next\_char}
 
 We know \AllChars, and \CP is a parameter \\
 \Output{A string representation of this layer}
 \Fn{compress\_layer() : String} {
 \ClassProportions $\gets$ proportion of each character class seen \\
 \tcp{e.x. \{"A-Z": 0.5, "a-z": 0.25, "1-9": 0.25\}}
 \MaxClass $\gets$ $\argmax(\ClassProportions)$ \\
 \MaxProportion $\gets$ \ClassProportions[\MaxClass] \\
 \If {$\MaxProportion > 0.95$} {
 	\Chars $\gets$ filter($x\rightarrow x \in \MaxClass$, \AllChars) \\
	\Histogram $\gets$ histogram(\Chars, bins=size(\MaxClass)) \\
 } \Else {
 	\Chars $\gets$ \AllChars
	\Histogram $\gets$ histogram(\Chars, bins=sum(size($\Class)\ \forall\ \Class \in \ClassProportions$)) \\
 }
 \If{ChiSquared(\Histogram) $> p$} {
	\KwRet{\MaxClass}
 } \Else {
 	\Captured $\gets$ 0 \\
	sort \AllChars by frequency \\
	\Representation$\gets$ [] \\
	\While{$\Captured < \CP$} {
		\Char$\gets$\AllChars.next() \\
		\Representation.add(\Char) \\
		\Captured$\gets$\Captured$+$frequency(\Char)\\
	}
	\KwRet{$``|"$.join(\Representation)}
 }
 }
\caption{Symbol layer representation of fitted tuples} 
\label{alg:layerrep}
\end{algorithm}

\subsubsection{Early Stopping} 
\label{subsec:earlystop}

When learning from structured data, it is common for
much of the computation time to go to fitting tuples
that do not contribute to the final \ds. Consider, for example, a
long list of well-formatted dates. After a few tuples, the
representation we are modeling will reflect the pattern, and will not change
as additional tuples are processed.

We can \emph{stop} learning when the model has converged and does not change after some
number of new tuples are consumed. To do this, we
 track how much the fit of new tuples changes during fitting. Initially
  the scores are expected to change a lot, they will
decrease and become steady over time -- especially when the
data is regular. The process of early stopping is inspired by the same
application of the Central Limit Theorem as in the representation
comparison. The early stopping process is shown in Algorithm~\ref{alg:branchfit}. 


\begin{algorithm}[h]
\scriptsize
\SetKwProg{Fn}{Function}{ }{end}
\SetKwData{DoneAdding}{done\_adding}
\SetKwData{Score}{score}
\SetKwData{AllScores}{all\_scores}
\SetKwData{Latest}{latest\_scores}
\SetKwData{Word}{word}
\SetKwData{StdScore}{current\_std}

\AllScores$\gets$[] \\
\Latest$\gets$[] \\
\Fn{needed\_sample\_size(\StdScore: float) : int} {
	\KwRet{int($(1.96 * x/0.1)^2$)} 
		\tcp{this is a normal distribution}
}

	\Fn{new\_word(\Word: String) : void} {
	
 \If{not \DoneAdding} {
 	\Score$\gets \sum_{1\leq i\leq |layers|} layers[i].add\_and\_output\_score(\Word[i])$ \\
	\Latest.append(\Score) \\
	\If{len(\AllScores) = 30 \tcp{Application of Central Limit Theorem}} {
		\Score$\gets$ avg(\Latest) \\
		\AllScores.append(\Score) \\
		\Latest$\gets$[]
	}
	\StdScore$\gets$ stdev(\AllScores) \\
	\If{len(\AllScores)$>$needed\_sample\_size(\StdScore)} {
		\DoneAdding$\gets$true
	}
  }
  }
\caption{Fitting tuples to branches}
\label{alg:branchfit}
\end{algorithm}

This technique allows us to skip large amounts of data while still finding good
approximate representations. The method fails when the attribute has
many different branches that are seen only later. For this reason, the technique
is disabled by default, and should be enabled when the user knows the data is
highly regular or randomly shuffled,  or if the user only desires a quick insight.

\subsection{Tuple Generation and Human Readability}
\label{subsec:representation}

A \ds needs to generate tuples that, though not necessarily part of
the given examples, conform to the domain of the examples  ($f(X)$,
formally). This is necessary for comparison, as we see in the next section.
Here, we explain how to generate tuples from a \ds
(\ref{subsec:generatingtuples}). Related to the generation of tuples is a string
representation of \ds which is readable by humans, a useful property to provide
an overview of the data to humans which we describe in \ref{subsec:readability}.


\subsubsection{Generating Tuples from a \ds}
\label{subsec:generatingtuples}

To generate a tuple, \system traverses the layers of the \ds bottom-up. It
generates \emph{characters} through its symbol layers. These are concatenated
into \emph{tokens} by the token layer, which also takes care of interleaving the
delimiters as necessary. Finally, tokens are concatenated into \emph{branches},
and the generator selects randomly the branch that would be chosen to generate
output a tuple.

To make sure each symbol layer generates characters leading to tuples that
represent the structure well, instead of returning the representation of its
character distribution, each symbol layer draws randomly from its corresponding
character distribution, producing a string from the symbol layer. For
convenience, the \textsf{compress\_layer} function of Algorithm
\ref{alg:layerrep} returns such representation.

\subsubsection{Making a \ds Readable}
\label{subsec:readability}

We want to serialize a \ds in a way that is easy to read, akin to how regexes
map the underlying DFA they represent to a string. The algorithm to achieve
this is similar to our tuple generation algorithm, but instead of specific
tuples, we want to output the general string representation that is represented
by \ds.

When traversing a \ds's layers bottom-up, we propagate partial representations
along the way. First, the symbol layers return either an individual character, a
character class (eg. \#, \textbackslash w, etc.), or a group of characters
depending on the result of the chi-squared test described in the previous
section. Then, all the symbol layer representations of a token representation
are appended, leading to a token, meaning that for a token representation $k_1$
with layers $l_1$ through $l_n$, where in the following $||$ represents the
concatenation operator:
$$str(k_1) = \bigparallel_{i=1}^n str(l_i)$$
These token representations are then interleaved with the appropriate
delimiters (kept during the learning process) to form branch
representations, given that:
$$str(b_1) = str(k_1)||h_1||str(k_2)||h_2\ldots$$ 
Finally, this is propagated upwards again, and the representation of an entire
\ds is simply an $OR$ of all of its branches:
$$R(X) = str(b_1) || ``|" || str(b_2)\ldots ``|" || str(b_n)$$
\subsection{Complexity and Expressiveness Analysis}

\notera{can we relate the results here to the points of 2.3?}

We analyze next the complexity of learning a \ds, performing branch-and-merge,
serializing the \ds as well as matching new strings. Note that by analyzing
our approach's complexity and expressiveness we can better understand how it
fits into the larger picture of techniques for information summarization and
extraction.

\textbf{Complexity:} Earlier, we showed that a \ds is a
DAG where each node represents a character distribution, internally
implemented as a set of linked lists of character symbols.
This representation supports fitness, comparison,
and generation algorithms.  Table \ref{table:complexityanalysis}
shows the time complexity of these algorithms in \system -- all algorithms in
\system are polynomial in the input size.  There are three
main algorithms: ``Scoring'', which assigns a fit score to a candidate word as
a function of how well it fits into an existing \ds, ``Branch and Merge'',
which samples a data column and decides how to contract or split the \ds when a
new sample is introduced, and ``Serialization'', which converts a \ds to a
human-readable and regex compatible notation.  The ``Scoring''
algorithm is used for both building a \ds, as well as matching a tuple against
an existing \ds, e.g., for outlier detection. 

\textbf{Expressiveness:} The ``Scoring'' and ``Branch and Merge'' algorithms
combined yield a data structure with the same expressiveness as
that of DAFSA.  Below we provide proofs of expressiveness
\ds w.r.t. to regular languages.

\begin{lemma}
A \ds can be converted in polynomial time to a DAFSA.
\end{lemma}

\begin{proof}
Since a \ds is a DAG where each node represents a character distribution, a DAFSA that accepts
all instances accepted by \ds can be trivially built in polynomial time via a BFS traversal of
the \ds. Nodes either accept a single character, or any character from a ``character-class''. 
Edges in the \ds are transitions in the resulting DAFSA.  Also note that this conversion to DAFSA
can be done in polynomial time because \ds itself is deterministic, e.g.\ the same string never
occupies more than one branch in the \ds.
\end{proof}

\begin{lemma}
A DAFSA can be converted in polynomial time to a \ds.
\end{lemma}

\begin{proof}
As in above proof: a \ds can be trivially built in polynomial time via a BFS traversal
of the DAFSA.
\end{proof}

\begin{theorem}
\ds expressiveness is equivalent to the set of regular languages that can be represented by DAFSA.
\end{theorem}

\begin{proof}
From the lemmas above, it follows that for every \ds there is at least one equivalent DAFSA, and
for every DAFSA there is at least one equivalent \ds.
\end{proof}

\begin{theorem}
\ds is equally as expressive as the finite regular languages, and is thus less expressive
than DFA.
\end{theorem}

\begin{proof}
Since \ds is equivalent to DAFSA, and DAFSA is less expressive than DFA, it follows that \ds is
necessarily less expressive than DFA.  Specifically, \ds cannot \emph{minimally} represent regular
languages that contains cycles.
\end{proof}

Note that we are not interested in learning minimal DFA in \system.  Indeed, even if we had chosen
a data structure that has the same expressiveness as that of DFA (e.g., it allows cycles), there
is no polynomial time algorithm guaranteed to produce a DFA of size at most polynomially larger than
the smallest consistent DFA using only positive samples \cite{minDFA}.

In practice, we also do not need to learn minimal DFAs here because our positive samples are drawn
from highly structured data, and instances of each language are finite, e.g., emails, telephone numbers,
and chemical identifiers. The expressiveness of DAFSA alone is quite powerful and covers all of our
finite languages use cases, while also doing a good job at situations where a dataset attribute is not
finite and the user only cares about tuples up to a certain size.  For example, assuming a dataset attribute
is captured by a small cyclic DFA, but we are only interested in instances of length at most $k$,
a DAFSA that represents this finite subset of the original language, and that is at most $k + 1$ states
larger than the DFA, can be obtained in polynomial time.


\begin{table}
\small
\centering
\begin{tabular}{|c|c|}
\hline
\textbf{Symbol} &
\textbf{Definition} \\
\hline
$W_d$ & Data column width (max tuple length). \\
\hline
$S_d$ & Number of items sampled from data column. \\
\hline
$B_x$ & Number of branches in the \ds. \\
\hline
$N_x$ & Number of nodes in the \ds. \\
\hline
\end{tabular}
\caption{List of symbols used in complexity analysis.}
\label{table:notation}
\end{table}

\begin{table}
\small
\centering
\begin{tabular}{|c|c|}
\hline
\textbf{Algorithm} &
\textbf{Time Complexity} \\
\hline
\emph{Scoring} & $O(M_x + N_x) \equiv O(B_x * W_d)$ \\
\hline
\emph{Branch and Merge} & $S_d * (scoring + B_x * W_d)$ \\
\hline
\emph{Serialization} & $O(M_x + N_x) \equiv O(B_x * W_d)$ \\
\hline
\end{tabular}
\caption{Time complexity for \system algorithms.}
\label{table:complexityanalysis}
\end{table}


\section{Comparing \ds{s}}
\label{sec:impl}

While there exist polynomial time algorithms for checking equivalence of
DFAs \cite{sipser_ToC} -- which implies that \ds can be checked for equality
against other DFAs in polynomial time -- we still require a distance measure.
We want to compare \ds{s}, so that we can identify columns with
syntactically similar values. We also wish to compare them with regexes, so that
we can propagate information associated to the regexes to the columns
represented by \ds. We explain how we achieve this in this section. We also
introduce a technique in \ref{subsec:lsh} to reduce the comparison complexity
and allow \system to be used in settings with large numbers of attributes.

\subsection{Measuring Similarity for Comparison}
\label{subsec:baseline}

The comparison operation of \system relies primarily on the scoring fit defined
in the previous section and the Central Limit Theorem, as we explain below.

\subsubsection{Comparing with other \ds{s}}
\label{subsec:otherstructures}

We want a \emph{syntactic distance} function between the structure represented by
different \ds, such as $D(X_1, X_2): \mathbb{R+} \times \mathbb{R+}$, that
returns a pair of scores between 0 and 1 representing how well the structure of each
\ds ``fits'' into the other.  Previously, in Section \ref{subsec:scoring}, we
discussed a \emph{scoring fit} obtained when fitting a tuple to a \ds. We
define now the fit of a \ds, $X_1$ \textit{into} $X_2$ as the average scoring
fit of the set of tuples represented by $X_1$ that fit $X_2$. 

In general, it is infeasible to generate \textit{all} possible tuples
represented by a \ds. Instead, we model $D(S, X_2)$ as a distribution for which
we want to estimate the mean fit with a certain degree of confidence. This
reduces the problem from one of generating all tuples, to one of generating a
subset of tuples that will allow us to estimate the mean fit in a statistically
significant manner. However, to reliably estimate the mean fit we would need
the underlying distribution of the data, which we do not know. We also do not
want to make assumptions about this distribution: it will be multi-modal at best,
and completely irregular at worst. 

To address this, we use the central limit theorem. Our idea is to sample
 the distribution in groups, taking sample averages. These sample averages
will approximate a normal distribution around $\mu$, the desired mean. Thus, in
order to compute the fit of $X_1$ in $X_2$, we generate groups of $n$ tuples
from $X_1$ and calculate their mean fit into $X_2$, as well as the standard
deviation of the approximately normal distribution. We use a confidence interval
of 95\% to estimate the sample size.

\subsubsection{Comparing with regexes}
\label{subsec:regex}

When comparing the structure represented by a \ds, $X$ with one represented by a
regex, $R$, we also want to obtain a tuple of scoring fits: how well $X$ fits
$R$ and the other way around. As with the comparison process between \ds{s}, our
approach involves generating tuples from the \ds (or regex), and then measuring
how well the generated tuples fit the regex (or \ds). The difference from the
approach in the previous section is that tuple values are binary, \ie either $X$
fits $R$ or it does not (and vice versa).  

To compute the similarity between a \ds a regex, we can
calculate the probability of the structure held in a \ds, $X$, fitting a regex,
$R$, as follows:
$$P(fit(X, R)) = \sum_{n=1}^{N} match(g(X, n), R) / N $$
where the function $g(X, n)$ generates a tuple from $X$ and the function $match$
returns 1 if a tuple fits $R$ and 0 if it does not. The total number of draws,
$N$ is chosen through standard application of the CLT, which allows us to treat
this as estimation of $\hat{P_{fit}}$, a Bernoulli random variable, and
therefore get an approximation within a certain range and confidence interval.

To compute the fitness of $R$ with respect to $X$, we use existing libraries
that produce strings from existing regular expressions, commonly known as
\emph{xeger}. Using one of these \emph{xeger}-like tools, we generate tuples
from the regular expression and then we apply the same technique in the opposite
direction.

We use this approach to compare \ds with already existing regular expressions
for our automatic label assignment application. We obtain good results that we
present in the evaluation section. However, it is worth noting a few limitations
of the approach with respect to the \ds-\ds comparison method. 

First, if a regex is too specific the similarity with a nearby structure may be
counter-intuitively low. For example, the structure of a regex that represents exactly
\emph{"ABCD"} will have a low similarity to a \ds's structure that represents
\emph{"ABCE"}, while this would not be the case if the two structures to be compared
would be represented by \ds{s}.

Second, due to our need to generate tuples from the regular expression,
the regex must be finite, so wildcard characters are not allowed. 
Although seemingly limiting, this is not a great disadvantage, as \emph{highly
structured} tuples will tend to lack wildcard characters -- which 
indicates a lack of structure.

\smallskip

\textbf{Why not compare regexes with the original data directly?} A natural
question is why do we compare a \ds with a regex instead of comparing the
original data directly to the regex. There are three key advantages to our
approach. First, it is easy to sample from a \ds, as it already represents the
branches in the underlying data. The alternative would be to perform expensive
random sampling on the data directly, which is difficult if we want to sample
from all the possible syntactical variations. Second, sampling from \ds involves
generating tuples in-memory and feeding them in streaming to the \ds, as opposed
to accessing and reading data from a data source. This is especially beneficial
because we need to repeate this operation every time a new regex is added to the
library, which happens often when multiple analysts participate in the process.
Last, it is more convenient to compare the \ds learned by \system to the regexes
as comparisons can naturally be parallelized, and once the \ds is learned, it is
readily available to be used with other applications.

\subsection{Efficient Large Scale Comparison}
\label{subsec:lsh}

Recall that one of our applications is to find which attributes are
syntactically similar. Naively, this entails performing an all-pairs comparison
of \ds, an $O(n^2)$ operation that becomes prohibitively expensive in settings
with large numbers of attributes. We rely on an approximate technique based on
locality-sensitive hashing (LSH) \cite{lsh} with minhash signatures. 

\mypar{A primer on LSH and minhash}To reduce the complexity of all-pairs
comparison from $O(n^2)$ to $O(n)$ we can model the problem as one of
approximate nearest neighbors (ANN). LSH solves ANN by using a family of hash
functions that bucket items -- \ie \ds of attributes in our case -- into the
same bucket when these are similar according to some similarity metric. When the
using Jaccard similarity, an effective method is minhash.  \cite{minhash}. With
the minhash method, the elements of a set are hashed with $K$ different hash
functions, and the minimum hash value for each $k$ function is chosen, leading
to the signature of the element.  The core insight of minhash is that the
probability of two minhashes being similar equals the Jaccard similarity.

\smallskip

Adapting LSH with minhash to \ds is challenging because we do not have sets of
elements, but \ds{s} that can generate them. The \ds, however, does not
generate sets of tuples deterministically, and the space of tuples it represents
can be very large, making it difficult to generate good minhash signatures. In
addition, instead of estimating the syntactic similarity of \ds we would be
just estimating the similarity of the sets of tuples they generate, which is not
what we want. For this method to work, we need a way of generating minhash
signatures from \ds{s} \emph{deterministically} and in a way that captures the
syntactic features learned during the building process.

We solve this by generating triples of the form (character, last\_hinge, index).
The first element represents the character or character class, the second one is
used to determine the token of which the character is part, and the last one
is the position of the character within the token. Codifying all this
information in triples preserves the structural information in a way that allows
us to still employ minhash. For example, for the string \emph{AB;CD}, we would
generate the set (A,0,0),(B,0,1),(C,1,0),(D,1,1)). With the set available, we
then use minhash to obtain a signature.

In our evaluation we show that this method  greatly reduces the comparison
runtime, with a minor reduction in accuracy.

\section{Evaluation}
\label{sec:evaluation}



\notera{R1.D2.3 - About Xtructure-Xtructure/Xtructure-Regex similarity measurements,
the number of generated tuples and the time cost of each Xtructure-Tuple
similarity meansurement are not provided}

\notera{R2.D4 - lots of other small things here}

In this section, we look at the performance of \system and study how it helps
address our motivating applications. Using a range of real datasets and
workloads we (1) study how \system can propagate labels from annotated regexes
to columns in the datasets (\ref{subsec:ala}); (2) use \system to learn \ds{s} on
columns of a dataset, and use these \ds{s} to identify syntactically similar
content (\ref{subsec:comparison}); and (3) use \system to detect syntactical
outliers from real data (\ref{subsec:outlier}). We also conduct a series of
microbenchmarks to understand the performance of \ds (\ref{subsec:micro}). 


\mypar{Datasets and setup} We use the following datasets: i) \textbf{university
data warehouse (DHW)} which consists of 161 tables and 1690 attributes with
information about departments within the university; ii) \textbf{ChEMBL (CHE)},
a public chemical database \cite{Gaulton2012ChEMBL} with 70 tables and 461
attributes; iii) \textbf{data.gov (GOV)}, US open government data, consisting of 2250 CSV files; and iv) \textbf{MassData (MAS)}, the open
government data from  Massachusetts, from which we use 10 attributes
for our outlier-detection experiment. For the
outlier-detection experiments we also use the KDDCUP99 and Forest Cover
datasets~\cite{uci}, which are standard datasets used in outlier
detection. Generally, we found that the results were 
robust to variation in the configuration parameters of \ds;  unless otherwise specified we set maximum branches to 3, the branching threshold to 0.1, and the capture threshold to 85\%.
For all the single-threaded experiments (all except as indicated), We use a
computer with a 1.7GHz Intel Core i7 and 8GB RAM.



\subsection{Automatic Label Assignment}
\label{subsec:ala}

To automatically label columns, we need (regex, label) pairs that associate
meaningful labels to the syntactic patterns described by the regexes. Given such
a library (which works as well as ground truth), we can use \system to learn 
syntactic patterns for each column in the database and compare these patterns
with the regexes in the library. Then, when we find a match, we assign the label
to the column. The quality of this application depends on the quality of our
comparison technique, which we evaluate here.

To obtain the library of (regex, labels), we manually assigned (regex, label)
pairs to more than 4,000 attributes from \textsf{DWH}, \textsf{CHEM} and \textsf{GOV}. The
specific number of attributes with assigned labels is shown in the ``\# total attrs.'' column
of table~\ref{table:labelassignmentresults}. The regexes are drawn from
regexlib.com~\cite{regexlib}, which has a collection of generic patterns. We
also added domain-specific regexes for chemical datasets. In both cases we choose
the most specific regex possible. For example, for an attribute containing even
numbers, we would use ``\textbackslash d\textbackslash d(0$|$2$|$4$|$6$|$8)''
rather than ``\textbackslash d\textbackslash d\textbackslash d''. 

We used \system to learn the \ds{s} for the 4000+ attributes and searched for
the nearest regex in the regex library, using the algorithm of
section~\ref{subsec:comparison}. We compared this nearest regex to the ground
truth regex we manually associated with each attribute. Table
\ref{table:labelassignmentresults} shows that we find over 94\% of correct
matches for the three datasets we use. This means that we can automatically
assign labels to 94\% of the data, which vastly reduces the human effort that
would otherwise be necessary. 

\begin{table}[h]
\small
\centering
\begin{tabular}{|c|c|c|c|c|}
\hline
\textbf{Dataset} &
\textbf{total attrs.} &
\textbf{correct matches} &
\textbf{\% matches} \\
\hline
\emph{DWH} & 1504 & 1417 & 94\%  \\
\hline
\emph{CHE} & 307 & 294 & 95.5\% \\
\hline
  \emph{GOV} & 2476 & 2355 & 94.9\% \\
\hline
\end{tabular}
  \caption{\ds-regex correct matches vs dataset.}
\label{table:labelassignmentresults}
\end{table}

Not all matches are equally useful. For example, we find matches of columns to
both InchI numbers and keys as well as to SMILES, and both of are annotated
with \emph{chemical id}. This vastly improves the discoverability of these
attributes, helping analysts with their tasks. In other cases, the match is with
a low specificity regex such as ``strings'' or ``numbers'', which although
correct is not insightful. This is an artifact of the quality of the (regex,
label) pairs we had available. In the enterprise scenario, we expect registries
of regexes built by domain experts to be of high quality, therefore leading to
good quality label annotation of the data.

In summary, this experiment shows that \textbf{\system is able to propagate 
labels to attributes for a wide range of attribute formats when a
library of (regex, label) is available}.



\subsection{Summarization and Comparison}
\label{subsec:comparison}

In this experiment, our goal is to use \system to find pairs of
similar attributes in a large dataset; such columns often represent duplicates,
or possible identifiers that can be used in joins. For this
application, we obtained ground-truth data consisting of pairs of similar
attributes from \textsf{CHE}. We collect all attributes within
the dataset whose name contains ``id'' (eg. ``tid,'' ``cell\_id,'' ``tax\_id'');
 attributes were removed and tuples shuffled in random order; a
volunteer labeled pairs of these shuffled nameless
attributes as syntactically similar, or different. We obtained labels for about
1000 pairs of attributes.  We learn \ds{s} for each attribute.

We evaluate the effectiveness of \system at finding  similar pairs.
First we perform an all-pairs comparison between the learned \ds{s} using the
method described in \ref{subsec:otherstructures}. This is $O(n^2)$ but  is an
intuitive method, useful when the number of attributes is small, or when we want to
quickly find all IDs in a database that are syntactically similar to one pre-selected
column ID. In this experiment, the method labels a  pair of columns as similar when
their similarity is above a given threshold, and then we measure precision and
recall of the results, which we show as the ``All Pairs'' line in
\F\ref{fig:prcurves_sim}. The figure shows a good accuracy, with the method
reaching an F1 score of around 0.82, and maintaining constant high precision
until a recall of about 0.8. 

\begin{figure}[t]
\centering
\vspace{-.2in}
\hspace{-0.5cm}
  \includegraphics[width=0.9\columnwidth]{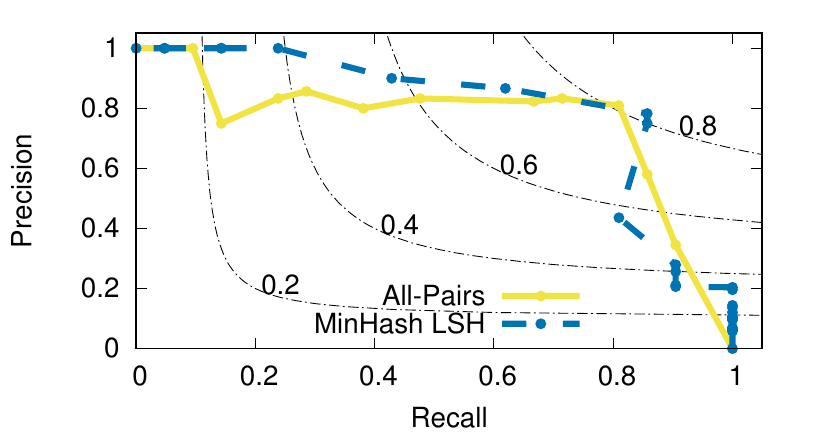}
\caption{PR curves for attribute similarity (ChEMBL)}
\label{fig:prcurves_sim}
\end{figure}

%
%

\mypar{Fast Comparison} As we have discussed, when the number of attributes is
in the thousands---which is often the case in large enterprises---an
all-pairs comparison becomes too expensive. To resolve this, we implemented a
variation of our method which uses the ``set signature'' approach described in
Section~\ref{subsec:lsh}. When using this approach, \ds{s} are clustered based
on the approximate Jaccard distance between their signature sets. These clusters
were then translated into pair labellings, giving an $O(n)$ time algorithm. The
precision and recall results for the same experiment using this method is shown
on the ``MinHash LSH'' line of \F\ref{fig:prcurves_sim}. The figure shows that
the quality of this alternative method is in fact similar to the all-pairs one,
with the curve shapes looking similar. The slight irregularities in the curve
(lack of smoothness) at high recalls are due to the cluster-based nature of the
LSH method, rather than direct comparison of each pair of attributes. This makes
sense because since we must pre-generate strings to generate the MinHash
signature, we make sure the strings uniformly represent the underlying data,
therefore increasing the signature quality. We further explore the details of
the performance tradeoff of these two methods in the microbenchmark in section
\ref{subsec:micro}.


\mypar{Qualitative Analysis} When the techniques yield errors, we found them to
be quite intuitive. For example, one common error we found was due to
irregularities in the data, such as two similar attributes not being detected because
one used ``nan'' to denote missing data, while the other used ``-1''. Another
common kind of error came from attributes with diverse representations and
implicit semantic meaning. For example, a human may label two attributes
containing variable length decimal numbers as different if the mean or standard
deviation of the numbers is different, which, in some cases \system is not able
to detect, yielding false positives.

%
%


\subsection{Syntax-Based Outlier Detection}
\label{subsec:outlier}

Next, we evaluate \system's ability to detect outliers, or anomalies, within
single attributes in a dataset. For this application, we use both the
\textsf{MAS} dataset (for which ground truth was manually collected through
volunteers), and several ``benchmark'' datasets in the field of outlier
detection, namely the KDDCUP 1999 intrusion detection dataset and the Forest
Cover dataset, both obtained from the UCI Dataset Repository~\cite{uci}. For
quantitative analysis, we use three of the outlier types from KDDCUP, as well as
the Forest Cover dataset; we then utilize the manually labeled \textsf{MAS} for
qualitative discussion. 

We learn a \ds per attribute from a subset of the tuples (with outliers present). We then freeze the \ds and
feed it new tuples, obtaining a fitness score which we use to find
outliers. The score is transformed into an \emph{outlier score} using a
weighted combination of the scoring fits (section~\ref{subsec:scoring}) of each
branch. We mark outlier scores that exceed an \emph{outlier threshold}. We present the
resulting PR curves in \F~\ref{fig:prcurves_outliers}.

\begin{figure}[t]
\vspace{-.2in}
\centering
\hspace{-0.5cm}
\includegraphics[width=0.9\columnwidth]{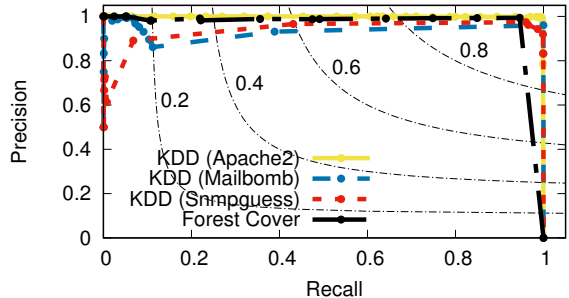}
\caption{PR curves for outlier detection}
\label{fig:prcurves_outliers}
\end{figure}

The figure shows the precision and recall for different values of the
\emph{outlier threshold}. The results show near-perfect performance on all four
of the large datasets. Since \system excels with large quantities of structured
data, we next perform a qualitative analysis of outlier detection using
\textsf{MAS}, a smaller but more complex real  dataset where outlier marking can actually be quite
subjective; this allows us to identify the areas where \system has the most
difficulty.

\mypar{Qualitative Analysis on MAS dataset} Many of the errors made by \system
are ambiguous to humans; in particular, the majority of the mistakes made were
in an attribute representing street address suffixes (ST, BL, AV, etc). The
source of ambiguity is lower-frequency, but still valid street suffixes, such as
``BL'' for boulevard, or ``PL'' for place, and whether or not \system marked
these as outlieries is simply a function of the aforementioned outlier threshold.
On the positive side, the system found outliers that were indeed errors in the
data, such as a ZIP code in Boston with more than 8 digits (the standard is 5
digits), or the tuple \emph{MIDNIGHT} among tuples representing hours as
digits.


%

\begin{figure}[!htb]
  \centering
  \subcaptionbox*{}[.7\linewidth][c]{
  \centering
    \includegraphics[width=.8\linewidth]{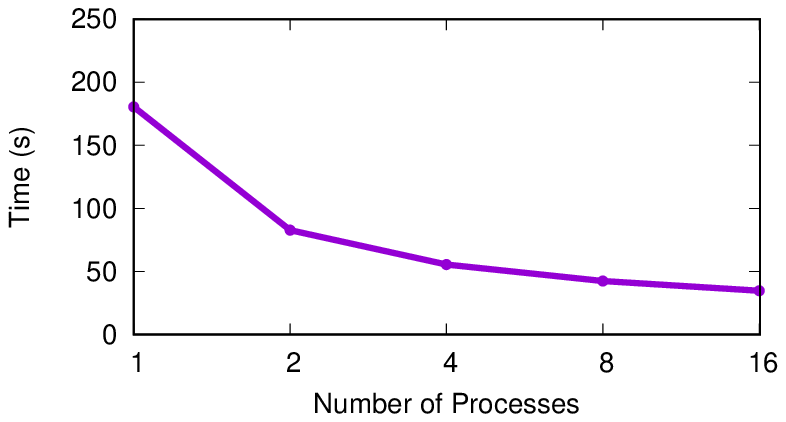}}\quad
  \subcaptionbox*{}[.7\linewidth][c]{%
  \centering
      \vspace{-.3in}
    \includegraphics[width=.8\linewidth]{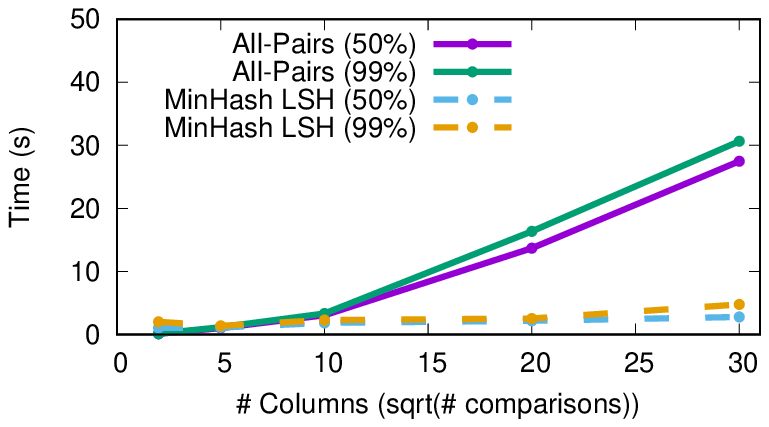}}\quad
    \vspace{-0.7cm}
  \caption{Top: Learning speed vs. (\# threads) on a 16-core machine, Bottom: All-Pairs vs MinHash LSH  methods}
  \label{fig:perfcomparison}
\end{figure}

\subsection{Microbenchmarks}
\label{subsec:micro}


\begin{figure*}[!htb]
\vspace{-.2in}
  \centering
  \subcaptionbox*{}[.3\linewidth][c]{
    \includegraphics[width=.3\linewidth]{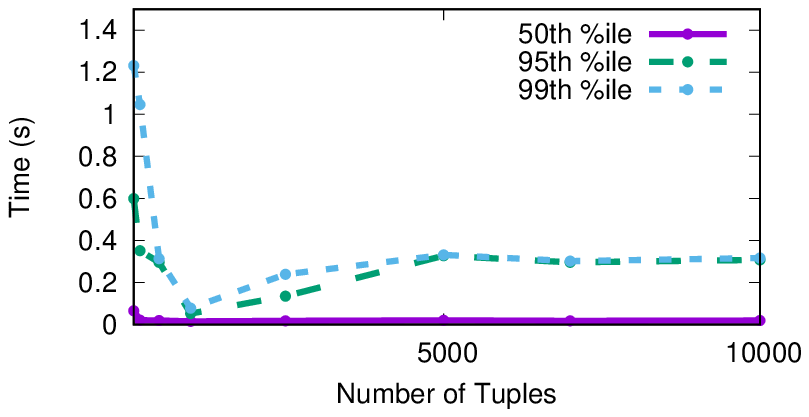}}\quad
  \subcaptionbox*{}[.3\linewidth][c]{%
    \hspace{-0.6cm}\includegraphics[width=.3\linewidth]{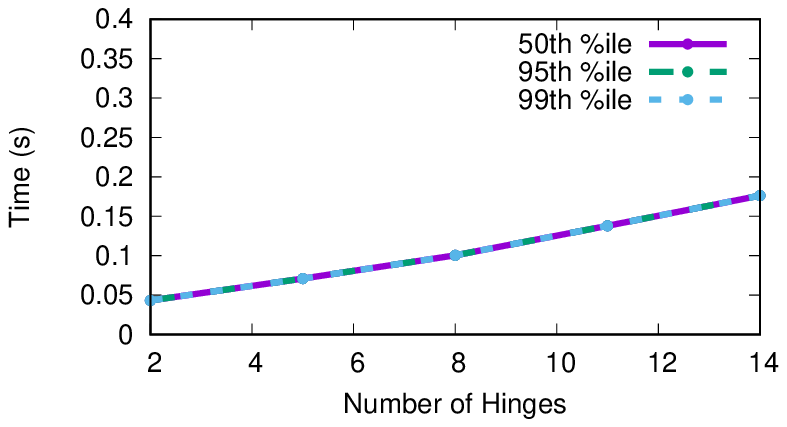}}
  \subcaptionbox*{}[.3\linewidth][c]{%
    \hspace{-0.5cm}\includegraphics[width=.3\linewidth]{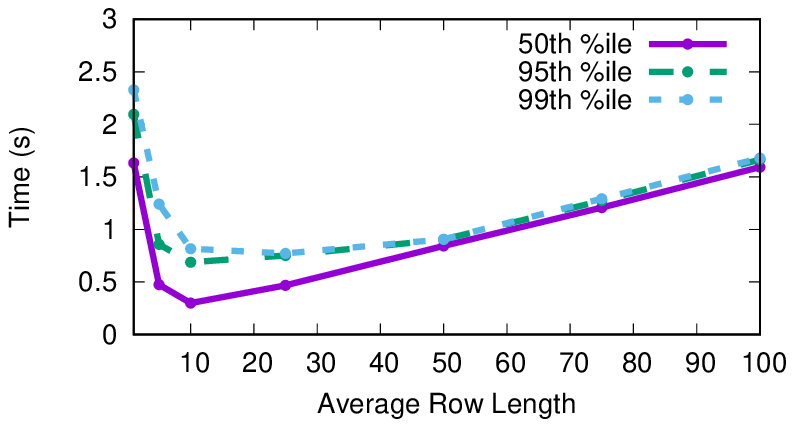}}\quad
  \vspace{-0.7cm}
  \caption{Performance micro-benchmarks using synthetic data sets.}
  \label{fig:perfmicro}
\end{figure*}


\vspace{.05in}
\noindent {\bf Learning Speed.} Properties such as number of tuples, number
of delimiters per tuple and pattern heterogeneity affect the performance of
\system. To measure these effects,  we generated synthetic data with varying properties and then
we ran \system on the data. \F\ref{fig:perfmicro}(a, b and c) shows the running
times averaged 10 times, with median, 95th and 99th percentile. 

In the first experiment (a), we use a fixed length data types (country currency
codes) and vary the number of tuples in the input dataset. In the second
experiment (b), we measure the impact of the number of hinges, which has an
effect on the number of tokens that \system must maintain. Here, we create input
datasets with variable number of hinges by concatenating YYYY-MM-DD formatted
dates, and fixed the number of tuples to 1000. Finally, in the third experiment
(c), we vary attribute value length using datasets with 1000 tuples and mixed
data types, which has an effect on the total number of branches that per \ds
that \system maintains. 

\system's performance in all 3 experiments grows linearly with the variable
of interest. Although absolute numbers are higher in the third experiment,
the overall runtimes are still small, with 99th percentile learning times
below ten seconds in all cases, and median generally times less than a second.

\vspace{.05in}
\noindent {\bf Parallel Scalability.} To understand the parallel scaling  of \system,
we generated data (about 20000 alphanumeric identifiers) and learned a \ds using
a different number of cores. We use \system with $max\_branches$ set to 7 and measure the time
 the learning process takes. Since the data is highly regular and \system 
only consumes a few samples before acheiving 95\% confidence and stopping early
(as described in section in~\ref{subsec:earlystop}), we disable early stopping
in order to accurately demonstrate the effects of parallelization.

\F\ref{fig:perfmicro} (d) shows our results. As expected, adding parallelism
reduces the total runtime up to the maximum number of hardware cores available
in the experimental machine. The system does not scale perfectly linearly after 8 cores
due to overheads during the merging stage of our algorithm, which could further
be reduced through optimizations, including hierarchical parallelization of the
merging operation itself.



\vspace{.05in}
\noindent {\bf All-Pairs vs Min-Hash Comparison.} Next we seek to understand the
difference in running time between the all-pairs and MinHash LSH methods for
\ds comparison, as was mentioned in attribute comparison
experiment in~\ref{subsec:comparison}. The experiment shows that in terms of results
the two methods perform very similarly in terms of accuracy on a real dataset;
in terms of runtime, we hypothesized that all-pairs method would be quadratic in
the number of attributes being compared pairwise, whereas MinHash LSH should be
linear. However, we also expected there to be lower overhead in the all-pairs
method for lower numbers of attributes. To verify this, we generate datasets of
uniform column length, but with varying numbers of columns. \ds{s} are learned
for each column (untimed), and then every pair of columns is compared using both
the all-pairs method and the MinHash LSH method. As \F~\ref{fig:perfcomparison} depicts, for
up to about 10 columns (which corresponds to on the order of 100 comparisons),
the all-pairs method out-performs MinHash LSH due to its low overhead, however
for larger numbers of attributes, the MinHash LSH method is superior in that it
scales linearly.

\noindent {\bf Invariance to Tuple Ordering.} Finally, we wish to measure the
invariance of \ds to the random shuffling of tuples. To do this, we take three
different synthetic attributes of various complexities, representing IP Address,
Title, and Latin Word. Each attribute contains 1000 tuples, and these are
shuffled in 20 distinct ways. Table~\ref{tab:shuffling} indicates the variation
across the unique shufflings, both in fitting time, and in the ``fitness score"
against the source column. The results show robustness against
bad orderings of tuples within an attribute.

\begin{table}[htbp]
\small
    \begin{tabular}{l|c|c|c}
	& IP Addresses & Titles & Latin Words \\ \hline
	Example & 159.112.55.237 & Dr. & cupiditate \\ \hline
	Mean (ms/line) & 0.27 & 0.23 & 1.0 \\
	Stdev (ms/line) & 0.07 & 0.05 & 0.3 \\
	Mean score & 0.18 & 0.25 & 0.41 \\
	Stdev score & 0.02 & 0.06 & 0.12 \\
    \end{tabular}
    \caption{Avg \& std.dev learning time, fit with random shuffling}
    \label{tab:shuffling}
\end{table}


\section{Related Work}
\label{sec:relatedwork}

In this section we discuss our contributions in the context of several
techniques and research areas related to \system.

\mypar{Information Extraction} \system is related to the vast literature of
information extraction (IE) in that it extracts a structural representation from
data. 
Most IE techniques have been proposed to extract structured
from totally unstructured data, such as text, or semi-structured data, such as
XML and HTML. In addition, most of those techniques require variable amount of
human input. Our approach differs in that it must work automatically and it
operates on structured data, producing one succint pattern that represents the
syntactic structured of a collection of input strings. We think of \system then,
as a technique related to IE that complements the existing techniques and
achieves good performance in important applications to large enterprises.

\mypar{Regular expression inference} This is the set of techniques most directly
related to \system; they can be used to extract syntactical patterns
from collections of strings. The most recent work uses multi-objective
optimization and aggressive space pruning to reduce the running time
\cite{Bartoli:2014:ASR:2780227.2780354, Bartoli:2016:IRE:2925263.2925390} of the
inference process.  Performance is still an issue for the method to be used in
enterprise settings as their evaluation shows---more than 40 min for learning a
dataset with 500 entries with 32 threads. In contrast, we have advocated and
demonstrated the better fit of \system, which reduces the unnecessary
expressiveness of regular expressions to gain in performance. \ds learns
patterns over a comparable data set on a single-core, single-thread in
significantly less time.

Other methods can be divided into whether negative examples are
required or not. Those that require negative examples are rarely suitable
in enterprise settings. Out of systems that only require positive examples,
\cite{Brauer:2011:EIE:2063576.2063763}, \cite{relie} and \cite{FERNAU2009521}
are the most relevant. With \cite{Brauer:2011:EIE:2063576.2063763} we
share our treatment of input characters as their character class (referred
to as token class in their case) to produce a higher level abstraction of input data.
Their method learns a cyclic DFA, while we have demonstrated that
this expressiveness is not necessary, and propose a more efficient learning algorithm.
Lastly, ReLIE \cite{relie} requires example regular expressions, that are then further refined. We
differ in that we operate without human input. Unlike all this work, we
focus on: i) designing \system to capture syntactical patterns in databases, and
not to solve the general -- and more complex -- problem of learning regular
expressions for infinite languages; ii) support efficient comparison of the learned patterns, which we
have shown helps in identifying syntactically similar content, automatically labelling data, and
identifying syntactic outliers.

\mypar{Program Synthesis}Program synthesis based methods have seen a surge in
popularity \cite{Lee:2016:SRE:2993236.2993244, Feser2015, flashextract}.
Unlike \system, their goal is typically to operate
and transform data, for example for data cleaning. This means that the
complexity of the structured they need to build and maintain internally is
higher than that of \system. For example, BlinkFill \cite{blinkfill} must build
an \emph{InputDataGraph} to then transform the data that is more expensive to
build than \system, and unnecessary for our goal. Other techniques, such as
\cite{Lee:2016:SRE:2993236.2993244, Feser2015} require negative examples and
differ again from our automatic technique.


\mypar{Other approaches} Last, approaches such as SystemT
~\cite{Krishnamurthy:2009:SSD:1519103.1519105} assume a human-in-the-loop,
infeasible in scenarios that require unattended operation. Potter's Wheel
\cite{potter} relies on dictionaries of predefined structures to identify data
errors and perform data transformations. Although we have demonstrated that
\system can be efficiently used to detect syntactic outliers, we do not share
goals and our methods to capture such patterns are different: both techniques
are orthogonal.

\section{Conclusions}
\label{sec:conclusions}


We presented \system, a method for learning syntactic structure of data as \ds.
While patterns in \system are less expressive than regexes, they are orders of
magnitude faster to learn, making it possible to learn patterns for thousands of
attributes which commonly occur in large databases.  Furthermore, \system performs
well at a number of database tasks, including automatic label assignment, column
summarization and comparison as well as outlier detection. 

\bibliographystyle{IEEEtran}

\balance

\bibliography{main}

\end{document}